\documentclass[a4paper,10pt]{article}
\usepackage{amsmath, amsfonts, amssymb, amsthm}
\usepackage[english]{babel}
\usepackage[applemac]{inputenc}
\usepackage{palatino}
\newtheorem{Definition}{Definition}
\newtheorem{Proposition}{Proposition}

\newtheorem{Lemma}{Lemma}

\begin{document}

\begin{titlepage}
\title{The Kowalewski's top revisited} 
\author{F. Magri\\
\small{Dipartimento di Matematica ed Applicazioni, Universita' di Milano Bicocca,}\\ \small{20125 Milano, Italy}}
\date{September 2, 2018}
\maketitle

\begin{abstract}
The paper is a commentary of one section of the celebrated paper by Sophie Kowalewski on the motion of a rigid body with a fixed point. Its purpose is to show that the results of Kowalewski may be recovered by using the separability conditions obtained by Tullio Levi Civita in 1904. 
\end{abstract}
\end{titlepage}

\section{Introduction}
The present paper is a commentary of Sec. 2 of the paper  \emph{Sur le problème de la rotation d' un corps solide autour d' un point fixe} written by Sophie Kowalewski in 1889 \cite{Kow1}. This celebrated paper  is composed by seven sections. Each contains a relevant result.
\par
Sec. 1 deals with the integrability of the Euler-Poisson equations  of  Mechanics, that is with the integrability of the equations of a rigid body with a fixed point in motion under the action of gravity. Starting from the remark that the classical cases of  Euler and Lagrange are solved by means of elliptic functions, which are meromorphic functions in the complex plane, Kowalewski sets the question of determining all the rigid bodies for which the solutions of the Euler-Poisson equations are meromorphic functions of the complex time. She finds a case escaped to Euler and Lagrange ( henceforth known as the Kowalewski's top ), and she claims that there are no other cases \cite{Kow2}. This claim raised the objections of Liapunov and Markov \cite{Liap}, due to some gap in the proof of Kowalewski , but presently it is commonly accepted as true. The original idea of Kowalewski has become one of the main technique to detect the integrability of the equations of motion of a dynamical system, and even accepted as one of the possible definions of the notion of integrable system \cite{ Aud}.
\par
In Sec. 2 Kowalewski demonstrates that the six equations of motion of her top may be reduced to a system of two first- order differential equations having the form of  Euler's equations for an hyperelliptic curve. In this section Kowalewski introduces the main tools needed to perform the reduction. They are:
\begin{itemize}
\item The quartic integral $k^2$ of Kowalewski.
\item The biquadratic function $R(x_1,x_2) $ of Kowalewski.
\item The fifth-order polynomial $R_1(s)$ of Kowalewski 
\end{itemize}
Let us explain intuitively their role. The integral $k^2$ ( combined with the classical integrals of the energy, of the vertical component of the angular momentum, and of the weight) serves to define a two-dimensional foliation. The motion of the top takes place on the leaves of this foliation, and therefore from an analytical point of wiew the top is a system with two degrees of freedom. The biquadratic function $R(x_1, x_2 ) $ serves to introduce a special system of coordinates $s_1$ and $s_2$ on the leaves of the  foliation. They are defined by 
\begin{equation*}
s := \frac{1}{2} w+ \frac{1}{2} l_{1} ,
\end{equation*}
where $3 l_1$ is the energy of the top, and $w$ is a solution of the quadratic equation
\begin{equation*}
w^2 -2 \frac{R(x_1, x_2)}{(x_1 - x_2)^2} w - \frac{R_1(x_1, x_2)}{(x_1 - x_2)^2} =0 
\end{equation*}
called  \emph{ the fundamental equation} by Golubev [7]. The second coefficient of this equation is the function of $R(x_1, x_2 ) $ defined by
\begin{equation*}
(x1-x2)^2 R_1(x_1, x_2) := R(x_1, x_1) R(x_2,x_2) - R(x_1, x_2 )^2 .
\end{equation*}
The fifth-order polynomial $R_1( s ) $ serves to write explicitly the reduced equations of motion on the leaves of the two-dimensional invariant foliation. The main result of Sec. 2 is the proof that in the coordinates $s_1$ and $s_2$  the Euler-Poisson equations of the Kowalewski's top assume the form 
\begin{eqnarray*}
0 &=& \frac{ds_1}{\sqrt{R_1(s_1)}}  + \frac{ds_2}{\sqrt{R_1(s_2)} }  \nonumber  \\
dt &=& \frac{s_1 ds_1}{\sqrt{R_1(s_1)} } + \frac{s_2 ds_2}{\sqrt{R_1(s_2)}}   ,
\end{eqnarray*}
that is the form of  Euler's equations related to the hyperelliptic curve $ y^2 = R_1(s) $.

In the remaining five sections, finally,  Kowalewski solves these equations and compute the angular velocity of the body as a function of the time in terms of  $\theta$-functions associated with the hyperelliptic curve. Her impressive computations were completed a few years later by Kotter \cite{Kotter}.
\par
\bigskip
As this brief description may suggest, the most critical point in the process of solution of Kowalewski is the choice of the coordinates $s_1$ and $s_2$.  It is really a difficult  problem to explain why the fundamental equation work so well. A noticeable advance in this direction was  a remark by A. Weil on an old theorem by Euler, concerning the solution of the Euler's equation on a general elliptic curve  \cite{Weil}. As noticed in \cite{HvM}, the fundamental equation  of Kowalewski coincides  with the " canonical equation " of Euler related to the elliptic curve $y^2 = R(x, x ) $ ( see also \cite{AudSil}). This remark leads to interpret  the change of variables of Kowalewski as an addition formula for elliptic functions, and points out an algebro-geometric origin of the fundamental equation. The algebro-geometric viewpoint has been widely developped in the past by M. Adler, L. Haine, P. van Moerbeke and collaborators, giving rise to the beautiful theory of the algebraically completely integrable systems ( see \cite{AvMVh} and the references collected therein).
\par
\bigskip
The viewpoint of the present paper is different. According to the analysis pursued in the paper, the origin of the fundamental equation of Kowalewski must be searched in the geometry of the leaves of the two-dimensional invariant foliation of the top. In the paper we prove that these leaves carry a remarkable differential geometric structure. We identify it and we prove that its occurrence is neither specific of the Kowalewski's top or occasional. It is a characteristic trait of  separable systems, which is implied  by the separability conditions discovered by Tullio Levi Civita in 1904 \cite{LC}. They claim that an Hamilton-Jacobi equation is separable in a set of  canonical coordinated $(u_j, v_j ) $ if and only if the Hamiltonian function verifies the equations
\begin{equation*}
\frac{\partial{H}} {\partial {u_{k}}} \frac{\partial{ H}}{\partial {u_{j}}}  \frac{\partial^2 {H}}{\partial {v_{k}} \partial {v_{j}} }
- \frac{ \partial{H}} {\partial{u_{k}}} \frac{\partial{H}}{ \partial{v_{j}}}  \frac{\partial^2{H}} {\partial {v_{k}}  \partial {u_{j}}} \\
-\frac{\partial{H}}{\partial{v_{k}}}\frac{\partial{H}}{\partial{u_{j}}}\frac{\partial^2{H}}{\partial{u_{k}}\partial{v_{j}}}
+\frac{\partial{H}}{\partial{v_{k}}} \frac {\partial{H}}{\partial{v_{j}}} \frac{\partial^2{H}}{\partial{u_{k}} \partial{u_{j}}}=0 .
\end{equation*}
In the last  section  we use the informations acquired by the study of these equations to construct the fundamental equation of Kowalewski in the case of the top.
\par

The paper is organized as follows. Sec. 2 is a brief survey of the second section of the work by Kowalewski. Sec. 3 is a first exploration of the differential geometry  of the leaves of the two-dimensional invariant foliation of the top. A threefold  geometric structure is identified and described by a 2-form, a tensor field of type $( 1, 1 ) $, and a metric. Its study leads to discover several new remarkable identities verified by the function $R(x_1,x_2) $ of Kowalewski. These identities are explained in the next two sections. In particular:  Sec. 4 presents a new criterion of separability, called the dual form of the Levi
Civita's conditions; Sec. 5 presents a new algorithm for the search of separation coordinates, called the method of Kowalewski's conditions.  When applied to the Kowalewski's top in Sec. 6, this method  quickly gives the fundamental equation of Kowalewski.

\section{Excerpts from the Kowalewski's paper}
In the second section of her paper,  Kowalewki consider the equations of motion of a rigid body with a fixed point under the action of gravity. She assumes that the principal moments of inertia satisfy the relations
\begin{center}
A = B = 2 C = 2 I   ,
\end{center}
and that the center of mass of the body belongs to the equatorial plane, that is to the plane orthogonal to the axis of material symmetry  passing through the fixed point. She considers as dynamical variables the three components $( p, q, r )$ of the angular velocity of the body, and the direction cosines $ ( \gamma, \gamma^{'}, \gamma^{''} ) $ of the invariable direction of the gravity with respect to the principal axes of inertia. Furthermore, she exploits the freedom in the choice of the inertial axes and in the choice of the units of measure to set $ I = 1 $, and to attribute the values
\begin{center}
$x_{0} = 1 $  \quad  $ y_{0} = 0 $  \quad  $ z_{0} =0 $  
\end{center}
to the coordinates of the center of mass in the moving reference frame. In place of these variables we prefer to use the components of the angular momentum
\begin{center}
$L1 = 2p I$   \quad   $L2 = 2 q I $    \quad   $L3 = r I $ 
\end{center}
and the components of the moment of the weight with respect to the fixed point
\begin{center}
$y1 = c_{0} \gamma $   \quad   $y2 = c_{0} \gamma^{'} $    \quad   $y3 = c_{0} \gamma ^{''} $ 
\end{center}
where $c_0 = Mgx_0$.  Accordingly we write the equations of motion of the rigid body  in the form
\begin{eqnarray}
\dot{L_1} & = & \frac{1}{2} L_2 L_3  \nonumber  \\
\dot{L_2} & = & -\frac{1}{2} ( L_3 L_1 + y_3 )   \nonumber  \\
\dot{L_3} & = & y_2    \\
\dot{y_1} & = & L_3 y_2 -\frac{1}{2}L_2 y_3   \nonumber  \\
\dot{y_2} & = & \frac{1}{2}L_1 y_3 - L_3 y_1   \nonumber   \\
\dot{y_3} & = &  \frac{1}{2} L_2 y_1 -\frac{1}{2} L_1 y_2  \nonumber .
\end{eqnarray}
A first property  worth of attention is the homogeneity  of the equations of motion. If one assigns degree 1 to the components of the angular momentum and degree 2 to the components of the moment of the weight, one can easily notice that 
\begin{displaymath}
deg(\dot{L_j})  =   2  \qquad deg(\dot{y_j})    =  3    \qquad        j= 1,2,3  .
\end{displaymath}
This means that the vector field defined by the equations of motion is homogeneous of degree 1. A second property concerns the integrals of motion. As noticed by Kowalewski, the above equations  admit four and not only three integrals of motion. These integrals may be chosen in many different ways. In the present paper we choose 
\begin{align}
c_1 & =  y_1^2+y_2^2+y_3^2    \nonumber  \\
c_2 & =  L_1y_1+L_2y_2+L_3y_3   \\
h_1 & =  \frac{1}{4}(L_1^2+L_2^2+ 2L_3^2)-y_1  \nonumber   \\ 
h_2 & = \frac{1}{32}(L_1^2-L_2^2+4y_1)^2+\frac{1}{32}(2L_1L_2+4y_2)^2 \nonumber  ,
\end{align}
and we notice that these integrals are related to those chosen by Kowalewski by the relations
\begin{displaymath}
h1= 3 l_{1}  \quad  h2 = \frac{1}{2} k^2  \quad  c_{2} = l  .
\end{displaymath}
For $c_{1}$  Kowalewski fixes the value $c_{1} = c_0^{2}$, while we keep its value arbitrary. This will allow us to readily control the degree of homogeneity of certain expressions containing the integrals of motion which will appear afterwards, by noticing that 
\begin{displaymath}
deg(h_{1}) = 2  \quad  deg(c_{2}) = 3  \quad  deg(h_{2}) = deg (c_{1}) = 4  .  
\end{displaymath}
\par
The final scope of Kowalewski in this section  is to infer the form of the equations of motion restricted to the level surfaces of the integrals of motion. To this end she starts to study a convenient parametrization of these surfaces. We skip this part of the paper, slightly departing from the line traced by Kowalewski, preferring to pass directly to the components of the dynamical vector field. We notice that the components $ \dot{L_1} $ and $ \dot{L_2} $ depend linearly on the coordinates $L_{3}$ and $y_{3}$, so that their squares depend linearly on $L_{3}^{2}$, $y_3L_{3}$, and $y_{3}^{2}$. On the level surfaces of the integrals of motion the last variables  are polynomial functions of $ (L_{1}, L_{2}, y_{1}, y_{2} )$, and therefore they can be easily eliminated. If one replaces the variables $y_{1}$ and $y_{2}$  by the variables ( of the same degree ) $e_1 =\frac{1}{4} (L1^2-L2^2)+  y_{1}$ and $e_2 = \frac{1}{2}  L1L2+ y_{2} $, the process of elimination gives:
\begin{eqnarray*}
8 \dot{L_1}^{2} &=& (4e_1- 2L_{1}^2 + 4h_{1} ) L_{2}^{2} \\
8 \dot{L_1}  \dot{L_2} &=& (4 e2L_2 - 4 h_{1} L_{1} - 4 c_{2})L_2 + L_{1}L_2 ( L_{1}^{2} -L_{2}^{2} ) )  \\
8 \dot{L_2}^{2} &=& - 4e_1 L_{2}^{2} +4 h_{1} L_{1}^{2} + 8 c_{2} L_{1} + 8 c_{1} -1/2 ( L_{1}^{2} - L_{2}^{2} )^{2} - 1/2 ( e_1^{2} + e_2^{2} )   .
\end{eqnarray*}
We stress that we have worked so far in the real domain.  Kowalewski works, instead, in the complex domain. Her complex variables are used almost universally in the literature on the Kowalewski's top, and a special role is attributed to them. I believe that this role has been overestimated, and my point of view is that one should privilege a tensorial viewpoint where all the coordinates are on the same footing. 
Nevertheless, it is mandatory at this point to pass to the Kowalewski's complex coordinates, in order to set the connection with her paper and to see the function $ R(x_1 , x_2 ) $ to appear . Kowalewski uses the variables $x_1 = \frac{1}{2} ( L_1 + L_2  i)  $ and $ x_2 = \frac{1}{2}( L_1 -  L_2  i) $, and the auxiliary functions $\xi_1 =  e_1 + e_2  i  $ and $  \xi_2 =  e_1 -  e_2  i $. Accordingly she obtains the expressions
\begin{eqnarray*}
-4 \dot{x_1}^{2} &=& R(x_1) + (x_1 - x_2)^{2} \xi_1  \nonumber  \\ 
4 \dot{x_1} \dot{x_2} &=& R(x_1 , x_2 )   \nonumber \\
-4 \dot{x_2}^{2} &=& R(x_2) + (x_1 - x_2)^{2} \xi_2  .
\end{eqnarray*}
In her notations the functions $R(x_1,x_2) $, $R( x_1)$, and $R( x_2) $ are defined by
\begin{eqnarray*}
R(x_1,x_2)   & = & -{x_1^{2}}{x_2^{2}}+2h_1x_1x_2+c_2(x_1+x_2)+c_1-2 h_2   \nonumber  \\
R( x_1) &=& -x_1^{4} +2h_1 x_1^{2} + 2 c_2 x_1 + c_1 -2 h_2   \nonumber \\
R( x_2) &=& -x_2^{4} +2h_1 x_2^{2} + 2 c_2 x_2 + c_1 - 2 h_2 .
\end{eqnarray*}
To eliminate the variables $ \xi_1 $ and $\xi_2$ from the equations of motion for $x_1$ and $x_2$, one  disposes of two relations: the first is $ \xi_1 \xi_2 = 2 h_2   $ ; the second follows from  $ \dot{x_1}^{2} \dot{x_2}^{2} - (\dot{x_1}  \dot{x_2})^{2} = 0 $. Together they imply the linear constraint
\begin{equation*}
R(x_2) \xi_1 + R(x_1) \xi_2 + R_1(x_1,x_2)  + 2 h_2   = 0
\end{equation*}
where the function  $R_1(x_1,x_2) $  appears for the first time. From $ \dot{x_1}^{2} \dot{x_2}^{2} - (\dot{x_1}  \dot{x_2})^{2} = 0 $ one readily sees that
\begin{equation*}
R_1(x_1,x_2) (x_1- x_2)^{2} =  R(x_1) R(x_2) - R(x_1,x_2)^{2} .
\end{equation*}
This equation emphasizes the prominent role of the function  $R(x_1,x_2)$, which generates all  the other relevant functions of the theory of the top. Therefore, to understand the paper of Kowalewski one has to clarify the specific role of this function.

The elimination of the variables $\xi_1$ and $\xi_2$  leads to a system of  implicit equations for $\dot {x_1} $ and $\dot {x_2} $ sufficiently complicated to give up  any hope of further progress. For this reason Kowalewski chooses another route, by introducing a new set of variables. More or less at the middle of the section she writes: 
\begin{flushleft}
Au lieu des deux variables $x_1$ et $x_2$ introduisons maintenant deux variables nouvelles $s_1$ et $s_2$, définies par les équations
\begin{eqnarray}
s_1 &=& \frac{R(x_1,x_2)- \sqrt{R(x_1)} \sqrt{R(x_2)} }{2 (x_1-x_2)^{2} } + \frac{1}{2}    l_1   \nonumber  \\
s_2 &=& \frac{R(x_1,x_2)+ \sqrt{R(x_1)}  \sqrt{R(x_2)} }{2 (x_1-x_2)^{2} } + \frac{1}{2} l_1  ;
\end{eqnarray}
$s_1$ et $s_2$ sont, comme on le voit, les deux racines de l' équation algébrique de second degré
\begin{equation}
(x_1-x_2)^2 ( s - \frac{1}{2} l_1) ^2  - R(x_1,x_2)(s-\frac{1}{2} l_1)-\frac{1}{4} R_1(x_1,x_2)=0 .  
\end{equation}
\end{flushleft}
She does not provide  motivations for her choice, but proceeds to prove a remarkable property of this change of variables.  Let $x_1$ and $x_2$ be two arbitrary functions of  the time, and let $\dot{x_1}$ and $ \dot{x_2}$ be their time derivatives. Let furthermore $S_1$ and $S_2$ be the functions of $s$ defined by
\begin{displaymath}
S_1 = 4 s_1^{3} - g_2 s_1 - g_3  \qquad   S_2 = 4 s_2^{3} - g_2 s_2 - g_3 .
\end{displaymath}
with   $ g_2 = k^{2} - c_{0}^{2} + 3 l_1^{2} $   and    $ g_3 = l_1 ( k^{2} - c_{0}^{2} -  l_1^{2} ) + l^{2} c_{0}^{2} $.   Then Kowalewski proves that 
\begin{eqnarray}
\frac{ds_1}{\sqrt{S_1}} &=& +\frac{dx_1}{\sqrt{R(x_1)}} + \frac{dx_2}{\sqrt{R(x_2)}}    \nonumber  \\
\frac{ds_2}{\sqrt{S_1}} &=&- \frac{dx_1}{\sqrt{R(x_1)}} + \frac{dx_2}{\sqrt{R(x_2)}}  .
\end{eqnarray}
This equation expresses a property of the Jacobian of the transformation from the coordinates $( x_1, x_2 ) $ to the coordinates $( s_1, s_2 ) $ and therefore, at the end, a property of the function $R(x_1,x_2)$ which defines the transformation. This property has attracted the attention of A. Weil and directed the further studies towards the algebro-geometric explanation of the change of variables of Kowalewski. In the last section we shall see that this property has a nice differential-geometric interpretation as well.

Joining together the two results obtained so far (the implicit equations for $\dot {x_1} $ and $\dot {x_2} $, and the property of the Jacobian of the transformation),  Kowalewski arrives to her main result. After some computations, she writes: 
\par
\begin{flushleft}
''  d' où il suit
\begin{eqnarray}
0 &=& \frac{ds_1}{\sqrt{R_1(s_1)}}  + \frac{ds_2}{\sqrt{R_1(s_2)} }  \nonumber  \\
dt &=& \frac{s_1 ds_1}{\sqrt{R_1(s_1)} } + \frac{s_2 ds_2}{\sqrt{R_1(s_2)}}   ,
\end{eqnarray}
$R_1(s)$ étant un polynome du cinquième degré, et les racines de l' équation $R_1(s) = 0 $ étant toutes différentes entre elles, les équations différentielles nous conduisent aux fonctions ultraelliptiques ''.
\end{flushleft}
It is this result of Kowalewski that we intend to analyse in the following sections. Before we remind that the same result has attracted the attention of many authors interested in the Lax representation of dynamical systems. Their goal is to recover the hyperelliptic curve $ y^2= R_1(s)$ as the spectral curve of a suitable Lax matrix with spectral parameter, in order to relate the method of  Kowalewski  to the process of linearization  of the flow on the Jacobian of the spectral curve explained in (\cite{Griff} ) and (\cite{Skl}  ). The Kowalewski's top has resisted this approach so far, and no  Lax matrix with the required property is known. Nevertheless the approach has been quite successful with some generalizations of the top, called the Kowalewski's tops in two-fields ( see (\cite{BBRS} ).

\section{Exploring the geometry of the Kowaleski's top}
In this section we start to look at the Kowalewski's top from the perspective of differential geometry. The aim is to identify three differential geometric structures attached to  the top that the algebro-geometric approach of Kowalewski has left in a cone of shadow. The search of these structures is guided by the study of the Levi Civita's conditions worked out in the next section. In this section we adopt an informal approach, letting the geometric structures emerge directly from the example. 
\par

The first step is to introduce into the picture a second vector field. It is the vector field  $ X_{h_2}$ associated with the quartic integral of motion discovered by Kowalewski. It is constructed by means of the Hamiltonian structure of the equations of motion of the rigid body ( see, for instance, [11]). Its equations are:
\begin{align}
L_1' & =  e_1 \dot{L_1}+e_2 \dot{L_2}    \nonumber  \\
L_2' & =   e_2 \dot{L_1}-e_1 \dot{L_2}    \nonumber  \\
L_3' & =  e_1 y_2 -  e_2 y_1   \nonumber \\
y_1' & =  \frac{1}{2}  y_3 ( e_1 L_2 -e_2 L_1 )   \\
y_2' & =   \frac{1}{2} y_3 ( e_1 L_1 +e_2 L_2 )   \nonumber  \\
y_3' & =   \frac{1}{2} y_1  ( e_2 L_1 -e_1 L_2 ) -\frac{1}{2} y_2  ( e_1 L_1+ e_2 L_2 ) \nonumber .
\end{align}
They show that $ X_{h_2}$ is a homogeneous vector field of degree $3$. The above equations have been written in a semi-explicit form in order to stress one point:  the first two components of  $X_{h_2}$  belong to the ideal generated by the first two components of  $X_{h_1}$. In other words, the components of the two vector fields satisfy two syzygies, and the coefficients of the syzygies are the functions $ e_1$ and $ e_2 $ already encountered before. We shall see, in the last section, that this seemingly odd property is relevant for the understanding of the  Kowalewski's paper. It is also worth to notice that this is the only point where the Hamiltonian structure of the equations of motion of the rigid body  plays a role. It will never be used afterwards. We stress this point to emphasize that the Hamiltonian scheme of the theory of integrable systems is less important for the solution of the equations of motion than commonly admitted. The basic scheme may be described as follows. A certain number of commuting vector fields are given on a manifold, together with a complementary number of functions which are constant along the vector fields. The level surfaces of the functions define a foliation to which the vector fields are tangent. If the leaves of the foliation carry the additional geometric structure described below the vector fields may be solved by separation of variables. There is no need of any Hamiltonian structure. We shall presently exhibit three remarkable geometric structures attached to the two-dimensional foliation of the Kowalewski's top.  They are described by a 2-form $\omega$, a tensor field $M$ of type $(1,1)$, and a metric $g$ respectively. All are related in a surprising way to the function $ R(x_1, x_2 ) $ of Kowalewski.

\paragraph{The 2 form $\omega$ }The search for the integrals of motion of a dynamical system, described by a vector field $X$, may be viewed as the search for the 1-forms $\alpha$ which annihilate $X$ and which admit an integrating factor:
\begin{equation*}
f \alpha = d I .
\end{equation*}
The potential $I$ of $\alpha$ is the sought integral. This process may be extended to foliations. Let us consider on the phase space of the top a class of homogeneous 2-forms $\omega$,  and let us look for those among them that annihilate the Kowalewski's foliation:  $\omega ( X_{h_1}, X_{h_2} ) = 0 $. The minimal degree of homogeneity worth of attention is $3$. The corresponding forms are:
\begin{equation*}
\omega = \sum_{j<k,l} a_{jk}^{l} L_{l} dL_j \wedge dL_k  + \sum_{jk} b_{jk} dL_j \wedge dy_k ,
\end{equation*}
where the coefficients $a_{jk}^{l}$ and $b_{jk} $ are constant. Evaluating these forms on the vector fields of Kowalewski, one readily checks that only four of them annihilate the foliation. Three are trivial and will not be specified.The fourth is
\begin{equation*}
\omega = \frac{1}{2}( L_1 dL_2 -L_2 dL_1 ) \wedge dL_3 +dL_2 \wedge dy_3 .
\end{equation*}
The main point is that this 2-form admits the integrating factor $f= L_2^{-2} $, so that
\begin{equation*}
 f  \omega  = d\alpha  .
\end{equation*}
The potential $\alpha$ is not uniquely defined, but a convenient choice is
\begin{equation} \label{eq :alpha}
\alpha = \frac{\dot{ L_1} dL_1 + \dot{L_2} dL_2 }{L_2^{2} }.
\end{equation}
This form is a homogeneous potential of $\omega$ that belongs to the span of $dL_1$ and $dL_2$. The existence of this 1-form is the first information provided by the differential-geometric analysis of the top.
\par
\bigskip
We shall now show that the 1-form $\alpha$ allows to discover a very subtle \emph{differential identity}  satisfied by the functions $R(x_1,x_2 )$ and $R_1(x_1,x_2 )$ used by Kowalewski to build her separation polynomial. To arrive to this identity we need to evaluate the components  of  $\alpha$ on the basis of vector fields $X_{h_1}$ and $X_{h_2}$. According to Eq. 8, they are 
\begin{eqnarray*}
\alpha (X_{h_1} ) &=& \frac{\dot{L_1}^{2}+\dot{L_2}^{2} }{L_2^{2}}    \\
\alpha (X_{h_2} ) &=& \frac{\dot{L_1} L_1'+\dot{L_2} L_2'}{L_2^{2}}  .
\end{eqnarray*}
If one evaluates these components in the coordinate system used by Kowalewski, 
one  discovers that 
\begin{eqnarray*}
\alpha (X_{h_1} ) &=& -\frac{R(x_1,x_2 )}{(x1-x2)^{2} }   \\
\alpha (X_{h_2} ) &=&- \frac{1}{2}\frac{R_1(x_1,x_2 )}{(x1-x2)^{2} }  +h_2 .
\end{eqnarray*}
So the 1-form $\alpha$ gives  the coefficients of the fundamental equation of Kowalewski up to some irrelevant additive constant. More importantly, it provides a new information on these coefficients. 

\begin{Lemma} \label{L1}The functions $R(x_1, x_2 )$ and $R_1(x_1,x_2 )$ of Kowalewski verify the  differential constraint
\begin{equation}
X_{h_1}(\frac{1}{2}\frac{R_1(x_1, x_2 )}{(x1-x2)^{2} }) - X_{h_2}(\frac{R(x_1, x_2 )}{(x1-x2)^{2} })  = 0 .
\end{equation}
\end{Lemma}
\begin{proof} Let us denote by $A=\alpha (X) )$ and $B=\alpha (Y)$ the values of $\alpha$ on any pair of commuting vector fields tangent to the two-dimensional foliation. Since $\alpha$ is the potential of a 2-form vanishing on the foliation $ d\alpha( X, Y) =0 $. Since the vector fields $X$ and $Y$ commute, the Palais formula for the exterior differential of 1-forms gives $ X(B) - Y(A)  = 0 $. The lemma is proved  by choosing  $X=X_{h_1} $ and $Y=X_{h_2} $.
\end{proof}

\paragraph{The tensor field M } The second remarkable geometric structure defined on the foliation of the Kowalewski's top is a tensor field of type $(1,1)$, henceforth denoted by $M$. It is generated by the fundamental equation of Kowalewski.
\par

To explain the definition of this tensor field, we need to make a digression on the relation between second-order polynomials and torsionless tensor fields. Let $(F, G) $ be an ordered pair of functions on the phase space of the top, which we identify as the coefficients of the second-order polynomial  $Q(u)= u^2 +F u +G $. Let moreover $M$ be a tensor field of type $(1,1)$ on the leaves of the Kowalewski's foliation, which we specify by its components on the basis $X_{h_1}$ and $X_{h_2}$ :
\begin{eqnarray*}
M X_{h_1}  &=&  m_1 X_{h_1} + m_2 X_{h_2}  \\
M X_{h_2}  &=&  m_3 X_{h_1} + m_4 X_{h_2} . 
\end{eqnarray*}
Let 's assume that  $\dot{F} G' - F' \dot{G} \ne 0 $ almost everywhere on the phase space of the top ( the symbols $\dot{F}$ and $F'$ denoting the derivatives of the function $F$ along the vector fields $X_{h_1}$ and $X_{h_2}$ respectively).  Then we claim that the functions $F$ and $G$ generate a torsionless tensor field $M$ on the leaves of the foliation.

\begin{Lemma} \label{L2}There is a unique torsionless tensor field $M$, on the leaves of the foliation of the top, which admits  the polynomial $Q(u)$ as its characteristic polynomial. Its components are given by 
\begin{eqnarray*}
m_1 &=& \frac{ \dot{G} G' +G \dot{F} F' -F \dot{F} G'}{ \dot{F} G' - F' \dot{G} }  \\
m_2 &=& \frac{ -\dot{G}^2 +F \dot{F} \dot{G} -G \dot{F} ^2}{ \dot{F} G' - F' \dot{G} }  \\
m_3 &=& \frac{ \dot{G} ^2 +F F' G' +G \dot{F}^2}{ \dot{F} G' - F' \dot{G} }  \\
m_4 &=& \frac{ -\dot{G} G' +F F' G' -G \dot{F} F'}{ \dot{F} G' - F' \dot{G} }  
\end{eqnarray*}
\end{Lemma}

\begin{proof} If one computes the Nijenhuis's torsion of the tensor field $M$ according to the standard procedure, one finds that $M$ is torsionless if and only if its components satisfy the  pair of differential equations
\begin{eqnarray*}
X_{h_1}(m_1m_4-m_2m_3)+m_2 (m_1'+m_4')-m_4( \dot{m_1}+\dot{m_4}) &=&  0  \\
 -X_{h_2}(m_1m_4-m_2m_3)+m_1(m_1'+m_4')-m_3( \dot{m_1}+\dot{m_4}) &=&  0   .
\end{eqnarray*}
By adding the conditions
\begin{align*}
m_1 + m_4 & = -F \\
m_1 m_4 - m_2 m_3 &= G ,
\end{align*}
which express that $Q(u)$ is the characteristic polynomial of $M$, one obtains a system of four linear algebraic equations for the coefficients of $M$. The determinant of the coefficients of the linear system is different from zero owing to the inequality satisfied by the functions $F$ and $G$. Hence the system has a unique solution. This solution is given by the formulas written above.
\end{proof}
\par
\bigskip
We use the recipe provided by this Lemma to work out the tensor field $M$ associated with the quadratic polynomial defining the fundamental equation of Kowalewski. The result is :
\begin{align*}
m_1 &=  \frac{R(x_1, x_2 )}{(x1-x2)^{2}} +\frac{1}{6} h_1  \\
m_2 &= \frac{1}{2}    \\
m_3 &= \frac{1}{2} \frac{R_1(x_1, x_2 )}{(x1-x2)^{2}}  \\
m_4 &= \frac{1}{6} h_1 . 
\end{align*}
Therefore also the tensor field $M$ is closely related to the coefficients $R(x_1, x_2)$ and $R_1(x_1,x_2)$ of the fundamental equation. What does it mean ? This question will be answered in Sec. 5. Here we simply show a noticeable consequence of this occurrence.

\begin{Proposition} The vector fields $X_{h_1}$ and $X_{h_2}$  commute with respect to the deformed commutator defined by the tensor field $M$ associated with the fundamental equation of Kowalewski :
\begin{equation*}
[ X_{h_1} , X_{h_2} ]_{M}= 0 .
\end{equation*}
\end{Proposition}
\begin{proof} Let us recall that every torsionless tensor field of type $(1,1)$ defines a deformed commutator on vector fields
\begin{equation*}
[ X , Y ]_{M} := [ MX , Y ]  + [ X , MY] - M [ X , Y ] .
\end{equation*}
By expanding the vector equation $[ X_{h_1}, X_{h_2} ] _{M} =0 $ on the basis  $X_{h_1}$ and $X_{h_2}$, one finds that the equation is true if and only if  the components of $M$ satisfy the first-order differential constraints
\begin{align*}
X_{h_1}(m_3)  &= X_{h_2}(m_1)   \\
X_{h_1}(m_4)  &= X_{h_2}(m_2) .
\end{align*}
The second equation is manifestly verified since the components $m_2$ and $m_4$  are constant. The first equation coincides with the differential constraint discovered in the study  of  the 1-form $\alpha$. This Proposition explains the geometric meaning of Lemma \ref{L1}.
\end{proof}

\paragraph{The metric $g$} The third remarkable geometric structure defined on the foliation of the Kowalewski's top is a metric $g$. In the coordinates $( L_1, L_2)$, it has the particularly nice form
\begin{equation}
g := \frac{dL_1\otimes dL_1 +dL_2\otimes dL_2}{L_2^{2}}   .
\end{equation}
It is characterized by the following two properties:
\begin{itemize}
\item  It maps the vector field $X_{h_1}$ into the 1-form $\alpha$ : $g( X_{h_1},  X ) = \alpha(X) $.
\item  It makes the tensor field $M$ symmetric : $g( MX, Y ) = g( MY, X ) $ .
\end{itemize}
It also displays a strict connection with the coefficients of the fundamental equation. Indeed its components on the basis $X_{h_1}$ and $X_{h_2}$ are
\begin{eqnarray*}
g ( X_{h_1} , X_{h_1} ) &=& -\frac{R(x_1,x_2 )}{(x1-x2)^{2} } \\
g ( X_{h_1} , X_{h_2} ) &=& - \frac{1}{2}\frac{R_1(x_1,x_2 )}{(x1-x2)^{2} }  +h_2  \\
g ( X_{h_2} , X_{h_2} ) &=& - 2 h_2 \frac{R(x_1,x_2 )}{(x1-x2)^{2} } .
\end{eqnarray*}
I guess that this metric has a relevant role in the theory of the Kowalewski's top. Probably it points out a connection with the theory of Killing tensors but this connection is unclear to me. Therefore, I restrict myself to signal the existence of this metric without elaborating on it.
\par
Summing up, one can say that the  investigation of the two-dimensional invariant foliation of the top according to the procedures of differential geometry reveals  a threefold unsuspected role of the basic function $R(x_1,x_2 )$ of Kowalewski.

\section{Levi Civita and beyond}
In this and the next section we develop a new point of view on the process of separation of variables. It was inspired by the outcomes of the previous investigation of the Kowalewski's top, and it has the aim of providing a coherent interpretation of these outcomes. Two are the major achievements supplied by the new point of view: the first is a new criterion of separability, called the dual form of the Levi Civita's conditions of 1904; the second is a new algorithm for the search of separation coordinates. It fits quite well with the work of Kowalewski, and for this reason it is  called the method of Kowalewski's conditions. In this section we deal with the criterion .
\par

To quickly explain the content of this criterion, let us consider a Liouville integrable system on a symplectic manifold of dimension $2 n$. The system is assigned  by giving $n$ independent functions $( H_1, \dots, H_n )$ which are in involution. These functions are known in a system of canonical coordinates $(x_j, y_j)$ fixed in advance. By assumption they do not verify the Levi Civita's separability conditions in these  coordinates. The problem of the search of separation coordinates is the problem of finding a new system of canonical coordinates
\begin{eqnarray*}
u_j &=& \Phi_j ( x_1,\dots,x_n,y_1,\dots, y_n )  \\
v_j &=& \Psi_j ( x_1,\dots,x_n,y_1,\dots, y_n )  
\end{eqnarray*}
such that the functions $H_a$  verify the Levi Civita's separability conditions when written in the new coordinates. To imply this property  
the functions $\Phi_j$ must be quite special. In this section we prove that these functions necessarily obey  a system of $n$ second-order differential constraints. These constraints have the form of $n$ polynomial equations on the first and second derivatives of the functions $\Phi_j$ along the Hamiltonian vector fields associated with the functions $H_a$. They are called the dual form of the Levi Civita's conditions because they work as the Levi Civita's conditions but in the opposite sense: while the Levi Civita's conditions demand to consider the derivatives of the functions $H_a$ with respect to the coordinates $u_j$, the new criterion of separability demands to consider the derivatives of the coordinates $u_j$ along the Hamiltonian vector fields associated with the functions $H_a$. The proof of the existence of the new differential constraints rests on the following geometric interpretation of the Levi Civita's separability conditions.

\begin{Proposition} Consider a separable Liouville integrable system, that is a Liouville integrable system plus a system of separation coordinates $(u_j, v_j)$. Assume that the functions $H_a$ depend on all the momenta $v_j$, so that the derivatives of the functions $H_a$ with respect to $v_j$ do not vanish. ( This assumption is common within the theory of Levi Civita's conditions). Hence, the Lagrangian foliation of the system is equipped with an infinite number of tensors fields $M$ of type $(1,1)$  which satisfy the following two conditions
\begin{itemize}
\item The Nijenhuis's torsion of $M$ vanishes.
\item The Hamiltonian vector fields associated with the functions $H_a$ commute in pair with respect to the deformed commutator defined by $M$ :
\begin{equation}
[ X_{H_a} , X_{H_b} ]_{M}= 0 .
\end{equation}
\end{itemize}
\end{Proposition}
\begin{proof} As is known, a set of function $H_a$ which are both in involution and separable ( in the sense that they satisfy the Levi Civita's separability conditions) are necessarily in separated involution. This means that they satisfy the stronger  involutivity conditions 
\begin{equation*}
\frac{\partial{H_a}}{\partial{u_j}}    \frac{\partial{H_b}}{\partial{v_j}} \\ 
- \frac{\partial{H_a}}{\partial{v_j}}   \frac{\partial{H_b}}{\partial{u_j}} = 0  
\end{equation*}
in the separation coordinates. This property has several important consequences. 
The main is that one can introduce the auxiliary functions
\begin{equation*}
\lambda_j := \frac{ \frac{\partial{H_a}}{ \partial{u_j}}}{\frac{\partial{H_a}}{\partial{v_j}}} 
\end{equation*}
which  do not depend on the choice of the function $H_a$. Related with them there are the vector fields
\begin{equation*}
X_j := \frac{\partial{}}{\partial{u_j}} - \lambda_j \frac{\partial{}}{\partial{v_j}} .
\end{equation*}
By the condition of separated involution, these vector fields are tangent to the leaves of the Lagrangian foliation. By the linear independence of the vector fields $\frac{\partial{}}{\partial{u_j}} $, they form a basis in the tangent space to the leaves of the  foliation. By the Levi Civita's conditions they commute. So the final result is that the leaves of the Lagrangian foliation associated with a separable Liouville integrable system admit a distinguished basis of commuting vector fields $X_j$. 
\par
\bigskip
Let us use this basis to define the tensor field $M$ by assigning its eigenvalues and its eigenvectors. We agree that $M$ admits the vector fields $X_j$ as eigenvectors, and that the corresponding eigenvalues $\mu_j$ are functions of the conjugate canonical coordinates $(u_j,v_j)$ associated with $X_j$. By this agreement 
\begin{displaymath}
X_k(\mu_j) = 0     \qquad  for \quad    k \neq j  .
\end{displaymath}
We claim that this tensor field satisfies the two conditions stated in the Lemma. The proof is a simple computation. For the first condition ( the vanishing of the torsion of $M$ ), one has to evaluate the torsion on the basis formed by the vector fields $X_j$ by using the equation $ M X_j = \mu_j X_j $. One readily recognizes that all the terms vanish owing to the assumption made on the eigenvalues. For the second condition, one preliminarly expands the Hamiltonian vector field $X_{H_a}$ on the basis $X_j$ :
\begin{equation*} 
X_{H_a }= \sum_{j} \frac{\partial{H_a}}{\partial{v_j}} X_j  .
\end{equation*}
Then one evaluates the deformed commutator $ [ X_{H_a} , X_{H_b} ]_{M} $. By using the spectral definition of the tensor field $M$ once again, one finds that the deformed commutator is the sum of three groups of terms: two of them vanish because of the condition of separated involution; the third vanishes because of the assumption on the eigenvalues.
\end{proof}
\par

Let us make three comments. The first concerns the arbitrariness in the definition of the tensor field $M$. One may notice that this arbitrariness reflects a similar degree of arbitrariness in the definition of the separation coordinates. Indeed, it is clear that any canonical transformation like
\begin {eqnarray*}
\bar{u_j} &=&\mu_j( u_j, v_j )  \\
\bar{v_j} &=& \nu_j( u_j, v_j ) 
\end{eqnarray*}
produces another system of separation coordinates. One may set up a 1:1 correspondence among separation coordinates and tensor fields $M$ by imposing the additional constraint 
\begin{equation}\label{aut}
\mu_j := u_j .
\end{equation}
Under this condition, the tensor field $M$ gives a faithfull representation of the coordinate system. This condition will be constantly enforced henceforth. The second comment is that the above result shows that the example of the Kowalewski's top is not isolated. All separable Liouville integrable systems share with the Kowalewski's top the property that the leaves of their Lagrangian foliation  carry  a tensor field $M$. The third comment is about the intrinsic character of the geometric conditions satisfied by the tensor field $M$. They are tensorial and can be written in an arbitrary coordinate system. This property is remarkable, since we started from the Levi Civita's conditions that are nontensorial. In a sense, the Proposition provides a tensorial formulation of the nontensorial Levi Civita's conditions. 
\par
\bigskip

The two conditions on $M$ ( the vanishing of the torsion and the vanishing of the deformed commutator ) impose severe restrictions on the eigenvalues of  $M$. The next step is to work out these restrictions explicitly. Although the argument which follows is completely general, we shall confine ourselves to the case $n=2$ for brevity,  and we agree to denote by the symbols $f_p$ and $f_s$ the derivatives of any function $f$ along the vector fields $X_{H_1}$ and $X_{H_2}$ spanning the two-dimensional invariant foliation.

As in the example of the top, we expand the tensor field $M$ on the basis of  the Hamiltonian vector fields 
\begin{eqnarray*}
M X_{H_1}  &=&  m_1 X_{H_1} + m_2 X_{H_2}  \\
M X_{H_2}  &=&  m_3 X_{H_1} + m_4 X_{H_2}  ,
\end{eqnarray*}
and we notice that the vanishing of the torsion of $M$  entails that the components $(m_1, m_2, m_3, m_4 )$ can be computed as functions both of the eigenvalues of $M$ and of the coefficients of its characteric polynomial. If we call $u_1$ and $u_2$ the eigenvalues of $M$, and $F$ and $G$ the coefficients of its characteristic polynomial, we find
\begin{align*} 
m_1 &= \frac{u_1u_{1p}u_{2s}-u_2u_{1s}u_{2p}}{u_{1p}u_{2s}-u_{1s}u_{2p}}   \nonumber \\
m_2 &= \frac{u_2u_{1p}u_{2p}-u_1u_{1p}u_{2p}}{u_{1p}u_{2s}-u_{1s}u_{2p}}    \\
m_3 &= \frac{u_1u_{1s}u_{2s}-u_2u_{1s}u_{2s}}{u_{1p}u_{2s}-u_{1s}u_{2p}}   \nonumber\\
m_4 &= \frac{u_2u_{1p}u_{2s}-u_1u_{1s}u_{2p}}{u_{1p}u_{2s}-u_{1s}u_{2p}}      \nonumber 
\end{align*}
and
\begin{eqnarray*} 
m_1 &=& \frac{ +Gp Gs +G Fp Fs -F Fp Gs}{ Fp Gs - Fs Gp }  \\
m_2 &=& \frac{ -Gp^2 +F Fp Gp -G Fp ^2}{ Fp Gs - Fs Gp  }  \\
m_3 &=& \frac{ +Gp ^2 +F Fs Gs +G Fp^2}{ Fp Gs - Fs Gp  }  \\
m_4 &=& \frac{ -Gp Gs +F Fs Gs -G Fp Fs}{ Fp Gs - Fs Gp  }  .
\end{eqnarray*}
The second representation has been proved in Lemma \ref{L2} ; the first representation follows from the second one, by setting $F = -(u_1+ u_2)$ and $G = u_1 u_2 $. 
\par

Let us now recall that the vanishing of the deformed commutator $[ X_{H_1}, X_{H_2} ] _{M} = 0 $ requires 
\begin{align*}
m_{3p}  -m_{1s} &= 0  \\
 m_{4p}  -m_{2s} &= 0   \nonumber    .
\end{align*}
By inserting the above representations of the coefficients of $M$  into these constraints, one obtains the desired conditions on the eigenvalues of $M$ and on the coefficients of its characteristic polynomial. We shall specify them in a moment. Before we interpret these conditions  from the standpoint of the theory of separation of variables. Imagine to consider a two-dimensional separable Liouville integrable system, defined by a pair of functions $H_1$ and $H_2$. Let us assume that the  separation coordinates are defined as the roots of a quadratic equations $Q(u)=u^2+ Fu +G =0 $, as in the example of Kowalewski. According to the geometric interpretation of the Levi Civita's separability conditions, the leaves of the Lagrangian foliation of the separable system is endowed with a tensor field $M$ whose eigenvalues are the separation coordinates and whose characteristic polynomial is $Q(u)$. This tensor field satisfies the two conditions studied above. Hence the separation coordinates verify the constraints following from these conditions. The conditions just obtained provides, therefore, a new criterion of separability, written either on the roots or on the coefficients of the polynomial which defines the separation coordinates.

\paragraph{\bf{Criterion}} The separation coordinates $u_1$ and $u_2$ of a two-dimensional separable Liouville integrable system satisfy the dual Levi Civita's separability conditions
\begin{align}
u_{1p}u_{2p} u_{1ss} -( u_{1p} u_{2s} +u_{1s}u_{2p} ) u_{1ps} +u_{1s}u_{2s} u_{1pp} & = 0   \\
u_{1p}u_{2p} u_{2ss} -( u_{1p} u_{2s} +u_{1s}u_{2p} ) u_{2ps} +u_{1s}u_{2s} u_{2pp} & = 0 \nonumber.
\end{align}
The coefficients of the separating polynomial of a two-dimensional separable Liouville integrable system satisfy the dual Levi Civita's separability conditions 
\begin{align}
& (-G_s^{2}+F F_s G_s-G F_s^{2}) F_{pp}  + ( 2G_p G_s - F F_p G_s- F F_s G_p+2 G F_p F_s) F_{ps} +   \nonumber  \\
&(-G_p^{2} +F F_p G_p + G F_p^{2} ) F_{ss}  =0 \\
&(-G_s^{2}+F F_s G_s-G F_s^{2}) G_{pp} +( 2G_p G_s - F F_p G_s- F F_s G_p+2 G F_p F_s) G_{ps} + \nonumber\\
&(-G_p^{2} +F F_p G_p + G F_p^{2} ) G_{ss}= ((F_p G_s - F_s G_p)^{2}  \nonumber . 
\end{align}
\par
\bigskip
As clarified by the argument leading to the criterion, the criterion has a large domain of application. For instance, the first equation is verified by the coordinates $s_1$ and $s_2 $ of Kowalewski, while the second equation is verified  by the coefficients of her fundamental equation. It can, therefore, be considered a suitable point of departure for the search of separation coordinates. Nevertheless, it  omits to answer two important questions : How does one solve these conditions ? Are these conditions sufficient as well ?
Both questions will be adressed in the next section. Before, we add the remark that there is a second approach to the dual Levi Civita's separability conditions which avoids the use of the Hamiltonian setting. It is based on the theory developped by Staeckel \cite{Stae}, but its exposition goes outides the limit of the present paper.

\section{ The method of Kowalewski's conditions}
The idea that inspires the method of Kowalewski's conditions is to lower the order of the differential equations satisfied by the separation coordinates. It rests on the remark that there is a special class of solutions of the dual Levi Civita's conditions which are selected by a system of \emph{first-order} differential equations. These first-order equations are the Kowalewski's conditions.
\par
Let us remind the geometrical setting. It consists of three elements: a two-dimensional foliation, the equations of the foliation ( that is a set of independent functions which are constant on the leaves of the foliation), and a pair of commuting vector fields tangent to the foliation.  To be definite, we assume that there are two commuting vector fields $X_1$ and $X_2$ and four functions $( H_1, H_2, H_3, H_4)$ as in example of the top. The vector fields are not required to be Hamiltonian.  The derivatives of the function $f$ along $X_1$ and $X_2$   are denoted by $f_p$ and $f_s$ as before.

\begin{Definition}{(\bf{Kowalewski's conditions })} Let $T$ and $D$ be two functions such that $ T_p D_s - T_s D_p \ne 0 $. They are said to verify the Kowalewski's conditions if
\begin{eqnarray*}
T_s - D_p &=& 0 \nonumber \\
D_s-T D_p + D T_p &=& 0 .
\end{eqnarray*}
Furthermore, let $A$ and $B$ be another pair of functions, with $A \ne 0 $. They are said to verify the auxiliary system attached to the solution $( T, D )$ of the Kowalewski's conditions if
\begin{eqnarray*}
A_s - B_p &=& 0 \nonumber \\
B_s-T B_p + D A_p &=& 0 .
\end{eqnarray*}
\end{Definition}

There are four reasons to consider the Kowalewski's conditions and the related auxiliary system. We now present them in the form of four separate claims, each dealing with a different aspect of the method of Kowalewski's conditions. To state them, we agree to denote by $Q(w) := w^2 + T w + D $ the quadratic polynomial associated with the solution $ (T, D)$ of the Kowalewski's conditions;  by $w_1$ and $w_2$ its roots; by $L$ the unique torsionless tensor field of type $(1,1)$, on the leaves of the two-dimensional foliation, having $Q(w)$ as characteristic polynomial. Before this tensor field was called $M$, but from now on , for clarity, we shall call $L$ the tensor field associated with the solutions of the Kowalewski's conditions, and $M$  the tensor field associated with the solutions  of the dual Levi Civita's conditions. 
\par
The first claim justifies our interest in the Kowalewski's conditions.
\begin{Proposition} Any solution of the Kowalewski's condition is a solution of the dual Levi Civita' s conditions. 
\end{Proposition}
\begin{proof}Consider the Kowalewski's conditions, and derive them along the vector fields $X_1$ and $X_2$. The result is a system of four polynomial relations among $T$, $D$ and their first- and second-order derivatives. Add these four relations to the Kowalewski's equations themselves. This gives a set of six polynomial relations. Consider now the dual Levi Civita's conditions on the coefficients $T$ and $D$ of the polynomial $Q(w)$, and regard them as a pair of polynomial relations on $T$, $D$, and their first- and second-order derivatives. The proof is completed by noticing that the last two polynomials belong to the ideal generated by the six polynomials engendered by the Kowalewski's conditions.
\end{proof}
The second claim points out the distinctive properties of the solutions of the Kowalewski's  conditions.
\begin{Proposition}\label{L} The roots $w_1$ and $w_2$, and the tensor field $L$ associated with any solution of the Kowalewski's condition enjoy the following properties:
\begin{itemize}
\item The roots $w_1$ and $w_2$  verify the first-order constraints (called the second form of the Kowalewski's conditions)
\begin{align*}
w_{1s} + w_2 w_{1p}  &=0 \\
w_{2s} + w_1 w_{2p}  &=0 .
\end{align*}
\item The roots $w_1$ and $w_2$ verify the dual Levi Civita' s conditions.
\item The tensor field $L$ has the simplified form
\begin{eqnarray*}
L X_1  &=& -T X_1 + X_2 \\
L X_2  &=& - D X_1    .
\end{eqnarray*}
\end{itemize}
\end{Proposition}
\begin{proof} The first property is proved by inserting the relations $ T= - ( w_1 + w_2 ) $ and $ D = w_1w_2 $ into the Kowalewski's conditions. The second property follows from the first property: it is proved by the same technique used to prove the previous Proposition. The third property is proved by evaluating the components of $L$ as rational functions of $T$, $D$, and their first derivatives.  Modulo the Kowalewski's conditions, these rational functions take the simplified form shown in the Lemma. 
\end{proof}
The third claim explains why the Kowalewski's conditions are interesting for the search of separation coordinates.
\begin{Proposition} The roots $w_1$ and $w_2$ associated with any solution $(T, D )$ of the Kowalewski's conditions  are separation coordinates for the vector fields $X_1$ and $X_2$ tangent to the foliations.
\end{Proposition}
\begin{proof}This Proposition has been proved in [17]. We repeat here quickly the argument leading to the conclusion. The conditions  $w_{1s} + w_2 w_{1p} =0 $ and $w_{2s} + w_1 w_{2p} =0 $ imply the validity of the  following expansions: 
\begin{eqnarray*}
\psi_1 \frac{\partial{}}{\partial{w_1}} &= & X_2 + w_1 X_1  \\
\psi_2 \frac{\partial{}}{\partial{w_2}} &= & X_2 + w_2 X_1  .
\end{eqnarray*}
The vector fields $X_1$ and $X_2$ commute as well as the vector fields  $\frac{\partial{}}{\partial{w_1}}$ and $\frac{\partial{}}{\partial{w_2}}$.  This property of commutativity entails that the function $\psi_1$ depends only on $w_1$, and that the function $\psi_2$ depends only on $w_2$. To close the proof it is sufficient to write the previous vector expansions in components.  For the vector field $X_1$ and  $X_2$  one obtains the equations
\begin{align*}
 \frac{w_{1p}}{\psi_1}  + \frac{w_{2p}}{ \psi_2 }  &= 0  \\
 \frac{w_1 w_{1p}}{\psi_1}  + \frac{w_2 w_{2p}}{ \psi_2 }  &= 1 ,
 \end{align*}
and 
 \begin{align*}
 \frac{w_{1s}}{\psi_1}  + \frac{w_{2s}}{ \psi_2 }  &= 1  \\
 \frac{w_1 w_{1s}}{\psi_1}  + \frac{w_2 w_{2s}}{ \psi_2 }  &= 0 
 \end{align*}
respectively. Assuming that $\psi_1$ and $\psi_2$ are algebraic functions of $w_1$ and $w_2$ respectively, one recognizes in these equations the Euler's equations associated with the algebraic curves defining $\psi_1$ and $\psi_2$. This is the mechanism of separation of variables induced by the Kowalewski's conditions. 
 \end{proof}

The fourth claim, finally, gives a geometrical interpretation of the Kowalewski's conditions.

\begin{Proposition} \label{K} Let $K$ be the unique tensor field of type  $(1,1)$ on the leaves of the two-dimensional foliation characterized by the following two properties:
\begin{itemize}
\item It maps the vector field $ X_1$ into $X_2$: $K X_1 = X_2  $ .
\item The functions $T$ and $D$, solutions of the Kowalewski's conditions, are the trace and the determinant of  $K$: $T:=trK $ and $D:=detK $:
\end{itemize}
Then the differentials of the trace $T$ and of the determinant $D$ satisfy the recursion relation
\begin{equation*}
K dT = dD . 
\end{equation*}
\end{Proposition}
\begin{proof} Consider the tensor field $ K = L + T Id  $. Since $-T$ and $D$ are the trace and the determinant of $L$  according to Lemma \ref{L}, it follows that $T$ and $D$ are the trace and the determinant of $K$ . Since $LX_1= -TX_1+X_2$ it is obvious that  $K X_1 = X_2  $. To prove the last property, let us evaluate the 1-form $KdT-dD$ on the vector fields $X_1$ and $X_2$.  One obtains:
\begin{align*}
KdT(X_1) -dD(X_1) &= dT(KX_1) -dD(X_1)= dT(X_2)-dD(X_1)\\
& = T_s - D_p  . \\
KdT(X_2) -dD(X_2) &= dT(KX_2) -dD(X_2)=dT(TX_2-DX_1)-dD(X_2)\\
& = T T_s - D Tp -D_s  .
\end{align*}
These identities prove the claim.
\end{proof}

To complete the discussion, we add ( without proof)  three more claims which show the role of the auxiliary system, and clarify the relations between the dual Levi Civita's conditions ( of the previous section) and the KowalewsKi's conditions (of this section).
\par
The first claim shows how to construct solutions of the dual LeviCivita's conditions from solutions of the Kowalewski's conditions.

\begin{Proposition} Let $(T,D)$ and $(A,B)$ be any solution of the Kowalewski's conditions and of the related auxiliary system. Then the functions 
\begin{align*}
F &= A T -2 B  \nonumber \\
G &= A^2 D - A B T + B^2 
\end{align*}
solve of the dual Levi Civita's conditions. Furthermore , any solution of the dual Levi Civita's conditions may be represented in this way.
\end{Proposition}
The second claim shows the opposite relation.

\begin{Proposition} Let $(F,G)$  be any solution of the dual Levi Civita's conditions, and let $(m_1, m_2 , m_3, m_4 ) $ be components of the tensor field $M$ associated with the quadratic polynomial $ Q(u)= u^2+Fu +G $. Assume that $m_2 \ne 0 $. Then the functions
\begin{align*}
A &= m_2 \\
B &= m_4 \\
T &= \frac{(m_4-m_1)}{m_2}  \\
D &=  -\frac{m_3}{m_2}
\end{align*}
are solutions of the Kowalewski's conditions and of the auxiliary system. Furthermore , any solution of the Kowalewski's conditions conditions may be obtained in this way.
\end{Proposition}

The third claim, finally, explains the relation between the roots of the polynomials $Q(u)=u^2+Fu+G$ and  $Q(w)=w^2+Tw+D$.
\begin{Proposition}Let $(F,G)$ and $(T,D,A,B)$ be two related solutions of the dual Levi Civita's conditions and of the Kowalewski's conditions:
\begin{align*}
F &= A T -2 B  \nonumber \\
G &= A^2 D - A B T + B^2 
\end{align*}
Then  the roots $u_1$ and $u_2$ are related to the roots $w_1$ and $w_2$ according to 
\begin{equation*}
u = A w + B .
\end{equation*}
Moreover: the roots $u_1$ is a function only of $w_1$; the root $u_2$ is a function only of $w_2$. Hence the roots $(u_1,u_2)$ are separation coordinates for the vector fields $X_1$ and $X_2$ as well. 
\end{Proposition}
It is worth to notice that according to the last Proposition no other restiction on the coordinates $u_1$ and $u_2$ is required in order to guarantee tha they are separation coordinates. Hence the dual Levi Civita's conditions are necessary and sufficients for two-dimensional invariant foliations.
\par
\bigskip
Collecting the informations splitted among the different claims, one arrives to the present algorithm for the construction of the separation coordinates of a pair of commuting vector fields $X_1$ and $X_2$ satisfying the assumptions of the geometric scheme adopted in this paper.

\paragraph{\bf{Algorithm}} To construct separation coordinates for the vector fields $X_1$ and $X_2$, the first step is to find a solution $(T,D)$ of the Kowalewski's conditions
\begin{eqnarray}
T_s - D_p &=& 0 \nonumber \\
D_s-T D_p + D T_p &=& 0 .
\end{eqnarray}
Then the roots $w_1$ and $w_2$ of  the quadratic equation
\begin{equation}
Q(w)=w^2+Tw+D =0
\end{equation}
are a first system of such coordinates. 
In this system of coordinates the equations of motion have the form
\begin{align*}
 \frac{w_{1p}}{\psi_1}  + \frac{w_{2p}}{ \psi_2 }  &= 0  \\
 \frac{w_1 w_{1p}}{\psi_1}  + \frac{w_2 w_{2p}}{ \psi_2 }  &= 1 ,
 \end{align*}
and 
 \begin{align*}
 \frac{w_{1s}}{\psi_1}  + \frac{w_{2s}}{ \psi_2 }  &= 1  \\
 \frac{w_1 w_{1s}}{\psi_1}  + \frac{w_2 w_{2s}}{ \psi_2 }  &= 0 
 \end{align*}
respectively. If the functions $\psi_1$ and $\psi_2$ are algebraic functions defined by  the polynomial equations
\begin{displaymath}
P_1( w_1 , \psi_1) = 0  \qquad P_2( w_2 , \psi_2) = 0
\end{displaymath}
the above equations of motion are Euler's type equations related to the algebraic curves defined by the above polynomials. Successively, one may look at the solutions of the auxiliary system
\begin{eqnarray}
A_s - B_p &=& 0 \nonumber \\
B_s-T B_p + D A_p &=& 0 .
\end{eqnarray}
associated to the solution found at the previous step. With any solution of the auxiliary system one may construct a new system of separation coordinates $u_1$ and $u_2$ according to
\begin{equation}
u = A w + B .
\end{equation}
The new coordinates are the roots of the quadratic equation
\begin{equation}
Q(u)=u^2+Fu+G =0
\end{equation}
whose coefficients are related to the solutions of the Kowalewski's conditions according to 
\begin{align}
F &= A T -2 B  \nonumber \\
G &= A^2 D - A B T + B^2 
\end{align}
\par
\bigskip
This algorithm will be applied, in the next section, to the Kowalewski' top. The aim is to show that the  few remarks on the geometry of the top done in Sec. 3 lead quickly to  the fundamental equation of Kowalewski.

\section{Rediscovering the coordinates of Kowalewski}
According to Propostions 5 and 6 of the previous section, the problem of constructing a pair of separation coordinates $w_1$ and $w_2$  for  the top is equivalent to the problem of constructing a tensor field $K$, on the leaves of the two-dimensional invariant foliation, such that
\begin{align*}
 K X_{h_1} &= X_{h_2 } \\
 K d(tr(K))  &= d(det(K)) ,
\end{align*}
We now show how to construct this tensor field. 
\par

We use two informations readily available from the study of the vector fields $X_{h_1}$ and $X_{h_2}$. The first  is provided by the syzygies
\begin{align*}
L_1' & =  e_1 \dot{L_1}+e_2 \dot{L_2}    \\
L_2' & =   e_2 \dot{L_1}-e_1 \dot{L_2}   
\end{align*}
They tell us that the tensor field $K$ defined by
\begin{align*}
K dL_1 & =  e_1 dL_1+e_2 dL_2    \\
K dL_2 & =   e_2 dL_1-e_1 dL_2    
\end{align*}
is symmetric and verifies the first condition: $K X_{h_1} = X_{h_2}  $ . The second  is provided by the 1-form $\alpha$. It tells us  that the functions
\begin{eqnarray*}
\alpha (X_{h_1} ) &=& \frac{\dot{L_1}^{2}+\dot{L_2}^{2} }{L_2^{2}}    \\
\alpha (X_{h_2} ) &=& \frac{\dot{L_1} L_1'+\dot{L_2} L_2'}{L_2^{2}}  .
\end{eqnarray*}
verify the first Kowalewski's condition
\begin{equation*}
X_{h_1}(\frac{\dot{L_1} L_1'+\dot{L_2} L_2'}{L_2^{2}} ) - X_{h_2}(\frac{\dot{L_1}^{2}+\dot{L_2}^{2} }{L_2^{2}} )  = 0 .
\end{equation*}
These informations can be elaborated as follows. First let us notice that the most general symmetric tensor field $K$ satisfying the condition $K X_{h_1} = X_{h_2}  $ is:
\begin{align*}
K dL_1 & =  e_1 dL_1+e_2 dL_2 + f ( \dot{L_2}^{2} dL_1 - \dot{L_1} \dot{L_2} dL_2  )   \\
K dL_2 & =   e_2 dL_1-e_1 dL_2 +f (- \dot{L_1} \dot{L_2} dL_1+ \dot{L_1}^{2} dL_2  )   ,
\end{align*}
where $f $ is an arbitrary function. Then let us notice that the trace and the determinant of this tensor field
\begin{align*}
T  &= f ( \dot{L_1}^2+ \dot{L_2}^2 )\\
D &= f ( \dot{L_1} L_1' + \dot{L_2} L_2' )-(e_1^2 +e_2^2) ,
\end{align*}
coincide with  the components of the 1-form $\alpha$ ( up to an irrelevant additive constant) if we choose $f  =  L_2^{-2}$. Let us make this choice. We know that the functions $T$ and $D$ verify already the first Kowalewski's condition. It remains only to check the second condition. This can be easily done. Once done, we can claim that the roots of the quadratic polynomial $Q(w) = w^2+T w + D $ are separation coordinates for the top, without computing the coordinates explicitly , and without writing the equations of motion in the new coordinates. We know this property \emph{a priori} as a consequence of the Kowalewski's conditions.

The study of the tensor field $K$ provides another interesting result. Let us pass to the coordinates $x_1$ and $x_2$ of Kowalewski, and let us represent the tensor field $K$ in these coordinates. The result is
\begin{align*}
K dx_1 & = - \frac{R(x_1, x_2)}{(x_1-x2)^2} dx_1+\frac{R(x_1)}{(x_1-x_2)^2} dx_2    \\
K dx_2 & =  \frac{R(x_2)}{(x_1-x_2)^2} dx_1-\frac{R(x_1, x_2)}{(x_1-x2)^2}  dx_2   . 
\end{align*}
This is a key formula. It reveals that the function $R(x_1, x_2 ) $  and its allied functions are the components of the tensor field $K$. Therefore, all these functions are united into a single geometric object. Moreover it becomes clear that the fundamental equation of Kowalewski is just the characteristic equation of the tensor field $K$. These outcomes are a concrete illustration of the claim that the separation coordinates are strictly related to a geometric structure possessed by the leaves of the invariant foliation of the top.

As a final remark, let us remember that the tensor field $K$ satisfies the condition: $KdT = dD $. This is a differential condition which must pass  to the function $R(x_1, x_2)$ defining the components of  $K$. One readily proves that this function must satisfy the partial differential equations
\begin{align*}
\frac{1}{2}  \frac{ \partial{ R(x_1,x_2) } } { \partial{x_2} } +\frac{1}{4} \frac{dR(x_1)}{dx_1} +R(x_1,x_2) -R(x_1)  =& 0 \\
\frac{1}{2}  \frac{ \partial{ R(x_1,x_2) } } { \partial{x_1} } +\frac{1}{4} \frac{dR(x_2)}{dx_2} -R(x_1,x_2)+R(x_1)  =& 0 
\end{align*}
They are the last identities we want to emphasize. They entail the basic identities (5), used by Kowalewski to linearize the flow of the top. These identities have thus received two different types of interpretations. From the algebro-geometric standpoint they are the trace of an addition formula for elliptic functions. From the differential geometric standpoint they are the condition  KdT = dD. This merging of  the two points of  view is an intersting fact that deserves further studies.

\bigskip
\noindent{\bf Acknowledgements}. This paper is dedicated to Emma Previato in the occasion of her 65th anniversary. I like to notice that Emma has studied  in Padova, as Tullio Levi Civita, and that she has done important works in the field of algebraic geometry, as Sophie Kowalewski.

\end{document}